\documentclass[11pt]{article}

\usepackage[margin=1in]{geometry}
\usepackage{amsmath,amssymb,amsthm}
\usepackage{braket}
\usepackage{microtype}
\usepackage[colorlinks=true,linkcolor=blue,citecolor=blue,urlcolor=blue]{hyperref}

\newtheorem{theorem}{Theorem}[section]
\newtheorem{lemma}[theorem]{Lemma}
\newtheorem{proposition}[theorem]{Proposition}
\newtheorem{corollary}[theorem]{Corollary}
\theoremstyle{definition}
\newtheorem{definition}[theorem]{Definition}
\newtheorem{conjecture}[theorem]{Conjecture}
\theoremstyle{remark}

\newcommand{\QE}{Q_{\mathrm E}}
\newcommand{\rv}{\operatorname{rv}}
\newcommand{\Rel}{\operatorname{Rel}}
\newcommand{\Span}{\operatorname{span}}
\newcommand{\supp}{\operatorname{supp}}
\newcommand{\E}{\mathbb E}
\newcommand{\1}{\mathbf 1}
\newcommand{\ip}[2]{\left\langle #1,#2\right\rangle}
\newcommand{\norm}[1]{\left\lVert #1\right\rVert}

\title{Two exact quantum queries use at most six variables}
\author{Yaqiao Li\footnote{Shenzhen University of Advanced Technology, Shenzhen, China, liyaqiao@suat-sz.edu.cn}}
\date{\today}

\begin{document}
\maketitle

\begin{abstract}
We give a direct analytic proof that every total Boolean function computable exactly with at most two quantum queries depends on at most six input variables. This bound is sharp. A previous derivation relied on a numerical classification of Boolean functions on at most four variables; to the best of our knowledge, ours is the first unconditional proof. The main difficulty beyond the one-query case is that query-path vectors need not be mutually orthogonal, so interference prevents the scalar weight-counting argument from applying. We overcome this difficulty by establishing a rigid support-and-branch decomposition: after the coordinates in the first-query support are fixed, the second query induces a family of exact one-query branches. This structural description may be of independent interest.
\end{abstract}

\section{Introduction}

Let $F:\{+1,-1\}^n\to\{+1,-1\}$ be a total Boolean function.  Its exact
quantum query complexity $\QE(F)$ is the minimum number of oracle queries used
by a quantum algorithm that outputs $F(x)$ with probability one on every
input.  This differs from the bounded-error model, where a fixed error
probability is allowed; see de Wolf's lecture notes~\cite{de2019quantum}.
Quantum query complexity is also a useful framework for understanding quantum
algorithms more generally~\cite{ambainis2018understanding}.

Write $[n]:=\{1,\ldots,n\}$ and let $x^{\oplus i}$ be obtained from $x$ by
negating $x_i$.  Coordinate $i$ is \emph{relevant} to $F$ if
$F(x)\ne F(x^{\oplus i})$ for some $x$.  Let $\Rel(F)$ be the set of relevant
coordinates and $\rv(F):=|\Rel(F)|$.  Define
\begin{equation}
  J_{\mathrm E}(k):=\max\{\rv(F): F\text{ is a total Boolean function and }
  \QE(F)\le k\}.
  \label{eq:def-J}
\end{equation}
Our goal is an analytic proof that $J_{\mathrm E}(2)\le 6$.

The polynomial method gives
$\deg(F)\le2\QE(F)$~\cite{BBCMW2001,de2019quantum}.  For $d=\deg(F)$,
Nisan and Szegedy~\cite{NisanSzegedy1994} proved
$\rv(F)\le d2^{d-1}$, while Wellens~\cite{Wellens2022} proved
$\rv(F)\le4.394\cdot2^d$.  Consequently,
\[
  J_{\mathrm E}(k)\le\min\{k,4.394\}\,4^k,
\]
which gives $4,32,192$ for $k=1,2,3$.  Moreover, the order $4^k$ cannot
be improved using only $\deg(F)\le2k$: Chiarelli, Hatami, and
Saks~\cite{ChiarelliHatamiSaks2020} constructed, for every $d$, a degree-$d$
Boolean function with $3\cdot2^{d-1}-2$ relevant variables.  

The one-query case is known.  Chen, Ye, and Li~\cite{ChenYeLi2020} showed that, up to permutations and negations of inputs and output, every nonconstant total function of exact query complexity one is either a literal or the parity of two variables.  Hence $J_{\mathrm E}(1)=2$.  
For two queries, Barnum, Saks, and Szegedy~\cite{BarnumSaksSzegedy2003} gave a semidefinite programming characterization.  
Montanaro, Jozsa, and Mitchison~\cite{MontanaroJozsaMitchison2015} used it to produce numerical optimal-success-probability tables for functions on at most four variables, while explicitly noting in ~\cite[Appendix A]{MontanaroJozsaMitchison2015} that ``An entry `1' means that the SDP solver claims a solution with success probability greater than 0.999; note that this does not strictly speaking imply the existence of an exact algorithm using that number of queries.'' 
Ye, Yang, and Huang~\cite{YeYangHuang2024} subsequently characterized the functions with $\QE(F)=2$, while explicitly noting that their result is conditioned on  ``Assuming the correctness of the numerical result for $n =4$'' from ~\cite{MontanaroJozsaMitchison2015}.  Their characterization implies $J_{\mathrm E}(2)\le 6$.  We prove the same bound without a numerical premise. Conceptually, we provide a support-and-branch structured analysis of two-query  quantum exact  algorithms.

\begin{theorem}
\label{thm:main}
$J_{\mathrm E}(2)\le 6$. 
\end{theorem}

In fact, the corresponding lower bound also holds, hence $J_{\mathrm E}(2) = 6$. Below we describe the lower bound construction due to Mukherjee
and Maitra~\cite{MukherjeeMaitra2021}. We mention that our proof structure resembles the structure of the lower bound construction.  

Mukherjee and Maitra~\cite{MukherjeeMaitra2021} constructed, for every $k\ge1$, a
Boolean function on $2^{k+1}-2$ relevant variables with exact query
complexity $k$.  Set $G_1(x_1,x_2):=x_1x_2$ where $x_i \in \{+1,-1\}$ are sign variables, and recursively define, using fresh selector
sign variables $u,v$ and disjoint inputs $x,y$,
\[
  G_k(u,v,x,y):=\frac{1+uv}{2}G_{k-1}(x)
       +\frac{1-uv}{2}G_{k-1}(y).
\]
In particular,
\begin{equation}
  G_2(x_1,\ldots,x_6):=\frac{1+x_1x_2}{2}x_3x_4
       +\frac{1-x_1x_2}{2}x_5x_6.
  \label{eq:sharp-example}
\end{equation}
The parity of two signs can be computed exactly with one quantum query by
Deutsch's algorithm.  Compute the selector $uv$, store its value coherently,
and use a controlled relabeling of the query register to run an exact
algorithm for the selected copy of $G_{k-1}$, without an intermediate
measurement.  Induction gives $\QE(G_k)\le k$.  On the other hand,
\[
  G_k=\frac{G_{k-1}(x)+G_{k-1}(y)}2
      +\frac{uv}{2}\bigl(G_{k-1}(x)-G_{k-1}(y)\bigr).
\]
The degree-$2k$ terms from the two disjoint copies cannot cancel, so
$\deg(G_k)=2k$.  The polynomial method gives $\QE(G_k)\ge k$, and hence
$\QE(G_k)=k$.  Every variable is relevant: a variable in either copy is
witnessed after selecting that copy, while $u$ and $v$ are witnessed by
choosing inputs to the two copies with opposite outputs.  Thus, if $N_k$ is the number of relevant variables,
\[
  N_1=2,\qquad N_k=2+2N_{k-1}=2^{k+1}-2,
\]
and therefore
\[
  J_{\mathrm E}(k)\ge2^{k+1}-2.
\]
The known cases $k=1,2$ attain this lower bound. This  motivates the following conjecture, which to the best of our knowledge has not appeared previously in the literature.

\begin{conjecture}
\label{conj}
For every positive integer $k$, $J_{\mathrm E}(k)=2^{k+1}-2$.
\end{conjecture}

For example, whether $\QE(F)\le3$ implies
$\rv(F)\le14$ is still open.  The study of Conjecture~\ref{conj} may therefore lead to a deeper understanding of exact constant-query algorithms.  

For the lower bound construction $G_2$, the first query has physical support only on $\{1,2\}$; its coherently stored parity routes the second query to $\{3,4\}$ or $\{5,6\}$.  Our proof recovers an analogous two-cell routing structure from any hypothetical exact two-query algorithm with at least seven relevant coordinates. Specifically, the upper bound proceeds by contradiction in two steps.  First, at least three relevant coordinates of a counterexample must lie outside the first-query support.  Second, fixing that support yields exact one-query branches.  We show that three outside coordinates force two branch types, whose cells are complementary short halfcubes; a degree-reduction argument then gives the count $2+2+2=6$.  Both steps require structural information about the states, besides the degree of $F$. 

\medskip
{\bf\noindent Acknowledgment:} The author thanks the reviewers for useful suggestions on improving the presentation.

{\bf\noindent Declaration of AI usage:} The analysis, proof ideas, and writing were partially developed through extensive interactions with AI tools. The author takes full responsibility for the contents of this article.

\section{Preliminaries}
\label{sec:prelim}

Throughout the paper, \(F:\{+1,-1\}^n\to\{+1,-1\}\) remains the total
sign-valued Boolean function of interest.  When a symbol denotes a
function, we use capital \(F,G,\ldots\) for sign-valued
Boolean functions, and lower-case  \(f,g,h,\varphi,\ldots\) denotes auxiliary scalar- or vector-valued maps, especially Hilbert-space-valued state maps.  
%Full-cube inputs are denoted by \(x\in\{+1,-1\}^n\); other letters are used when the input is split into blocks, a promised domain is parameterized, or an auxiliary cube variable is introduced.  
We retain the
notation \([n]\), \(x^{\oplus i}\), \(\Rel(G)\), and \(\rv(G)\) from
the introduction.  A pair \((x,x^{\oplus i})\) on which the function values
differ is called a
\emph{bichromatic \(i\)-edge}.  We use \(\1\) for an all-ones input of
the dimension determined by context.

\subsection{Boolean functions and Fourier analysis notation}

For \(A\subseteq[n]\), define the Fourier character
\[
    \chi_A(x):=\prod_{i\in A}x_i,
    \qquad \chi_\varnothing(x):=1.
\]
The characters form an orthonormal basis for real-valued functions on the
Boolean cube under the Fourier inner product
\[
    \langle f,g\rangle
    :=\E_x[f(x)g(x)],
\]
where the expectation is uniform.  Hence every \(g:\{+1,-1\}^n\to
\mathbb R\) has the unique expansion
\begin{equation}
    g(x)=\sum_{A\subseteq[n]}\widehat g(A)\chi_A(x),
    \qquad
    \widehat g(A)=\E_x[g(x)\chi_A(x)].
    \label{eq:fourier-expansion}
\end{equation}
Its multilinear degree is
    $\deg(g):=\max\{|A|:\widehat g(A)\neq0\}$.
For a vector-valued function \(h:\{+1,-1\}^n\to V\), where \(V\) is a
real or complex vector space, the same definitions apply. 
 
Fourier expansion is
unique, and multiplication by a coordinate has the explicit effect
\begin{equation}
    x_i\chi_A(x)=\chi_{A\mathbin\triangle\{i\}}(x),
    \qquad
    x_i g(x)
    =\sum_{B\subseteq[n]}
       \widehat g(B\mathbin\triangle\{i\})\chi_B(x).
    \label{eq:fourier-toggle}
\end{equation}
Here \(\mathbin\triangle\) denotes symmetric difference.  Equivalently,
    $\widehat{x_i g}(B)
    =\widehat g(B\mathbin\triangle\{i\})$.

Finally, we need the degree lower bound ~\cite{BBCMW2001}
\begin{equation}
    \deg(F)\leq2\QE(F).
    \label{eq:polynomial-method}
\end{equation}

\subsection{The exact query model and the direct two-query form}

Let \(\mathcal H\) be a complex Hilbert space.     Adjoin a
dummy query label \(0\) and set \(x_0:=1\).  In the standard phase-oracle
model, \(\mathcal H\) has an orthogonal decomposition
\(R_0,\ldots,R_n\) of the identity and the query is
\[
    O_x^R:=\sum_{i=0}^n x_iR_i, 
    \qquad
    R_iR_j=\delta_{ij}R_i,
    \qquad
    \sum_{i=0}^nR_i=I.
\]
Here \(\delta_{ij}\) is the Kronecker delta.    The dummy
sector \(R_0\) is unchanged by the oracle.

A \(T\)-query algorithm has final state
\[
    U_TO_x^RU_{T-1}\cdots U_1O_x^RU_0\ket0,
\]
where the \(U_i\) are input-independent unitaries.  The exact quantum
query complexity \(\QE(F)\) is the least \(T\) for which a final measurement of
such a state returns \(F(x)\) with certainty on every input \(x\).

A standard two-query algorithm therefore has final state
\[
    U_2O_x^RU_1O_x^RU_0\ket0,
\]
where the \(U_i\) are input-independent unitaries.  Put
\[
    \ket\phi:=U_1U_0\ket0,
    \qquad
    P_i:=R_i,
    \qquad
    Q_j:=U_1R_jU_1^*.
\]
Then, \(P_0,\ldots,P_n\) and \(Q_0,\ldots,Q_n\) are orthogonal
decompositions of the identity.  Define
\begin{equation}
    O_x^P:=\sum_{i=0}^n x_iP_i,
    \qquad
    O_x^Q:=\sum_{j=0}^n x_jQ_j.
    \label{eq:direct-oracles}
\end{equation}
Since \(O_x^Q=U_1O_x^RU_1^*\),
\[
    O_x^PO_x^Q\ket\phi
    =O_x^RU_1O_x^RU_0\ket0.
\]
The remaining unitary \(U_2\) may be absorbed into the final measurement.
Thus, every two-query algorithm can be studied in the direct form
\begin{align}
    \ket{\varphi_1(x)}
    &:=O_x^Q\ket\phi
      =\sum_{j=0}^n x_jQ_j\ket\phi,
    \label{eq:first-state}\\
    \ket{\varphi_2(x)}
    &:=O_x^P\ket{\varphi_1(x)}
      =O_x^PO_x^Q\ket\phi.
    \label{eq:second-state}
\end{align}
%Accordingly, \(Q\) is always the first-query decomposition and \(P\) the second-query decomposition.  A one-query algorithm is written \(O_x^Q\ket\phi\) after the same absorption.  An algorithm using fewer than two queries may be padded by a query supported entirely on the dummy label.

By a standard dilation, the final two-outcome measurement may be assumed
projective.  Let \(\mathcal H_+\) and \(\mathcal H_-\) be its two
orthogonal output subspaces after the final unitary has been absorbed.
For \(\tau\in\{+1,-1\}\), set
\[
    \mathcal H_\tau
    :=\begin{cases}
       \mathcal H_+,&\tau=+1,\\
       \mathcal H_-,&\tau=-1.
      \end{cases}
\]
Let \(\mathsf M\) be the self-adjoint unitary acting as \(+I\) on
\(\mathcal H_+\) and as \(-I\) on \(\mathcal H_-\).  Exactness is the
pointwise identity
\begin{equation}
    \mathsf M\ket{\varphi_2(x)}
    =F(x)\ket{\varphi_2(x)}
    \qquad(x\in\{+1,-1\}^n).
    \label{eq:exactness}
\end{equation}
In particular, final states belonging to inputs with different function
values are orthogonal.

For \(0\leq i,j\leq n\), define the \emph{query-path vectors} and their
weights by
\begin{equation}
    v_{ij}:=P_iQ_j\ket\phi,
    \qquad
    a_{ij}:=\norm{v_{ij}}^2.
    \label{eq:path-vectors}
\end{equation}
Then, 
\begin{equation}
    \ket{\varphi_2(x)}
    = O_x^PO_x^Q\ket\phi
    =\sum_{i,j=0}^n x_ix_jv_{ij}.
    \label{eq:path-expansion}
\end{equation}
This is a vector-valued function of degree at most two.
Because \(\sum_iP_i=I\) and \(\sum_jQ_j=I\),
\begin{equation}
    \sum_{i,j=0}^n a_{ij}=1.
    \label{eq:path-normalization}
\end{equation}
For each first-query label \(j\), put
\begin{equation}
    q_j:=\sum_{i=0}^na_{ij}
        =\norm{Q_j\ket\phi}^2.
    \label{eq:first-query-weight}
\end{equation}
The \emph{first-query support} is the set of physical
coordinates
\begin{equation}
    S:=\{j\in[n]:q_j>0\}.
    \label{eq:first-query-support}
\end{equation}
The dummy label \(0\) is not a physical coordinate and is not included in
\(S\).  If \(j\notin S\), then \(Q_j\ket\phi=0\), equivalently
\(v_{ij}=0\) for every \(i\).

\section{Preparation}
\label{sec:preparation}

\subsection{Exact one-query functions and degree-one states}

In the notation of Section~\ref{sec:prelim}, a normalized one-query state
is the vector-valued function
\begin{equation}
    f:\{+1,-1\}^n\longrightarrow\mathcal H,
    \qquad
    f(x):=O_x^Q\ket\phi
         =f_0+\sum_{i=1}^n x_if_i,
    \qquad
    f_i:=Q_i\ket\phi.
    \label{eq:one-query-state}
\end{equation}
Here \(f_0\) is the dummy-sector component.  Since the \(Q_i\) are
mutually orthogonal projections,
\begin{equation}
    \ip{f_i}{f_j}=0\quad(i\neq j),
    \qquad
    \sum_{i=0}^n\norm{f_i}^2=1.
    \label{eq:one-query-sector-orthogonality}
\end{equation}

The following lemma has been proved in Chen, Ye, and
Li~\cite{ChenYeLi2020}. We give a proof below for completeness since it is also very short. Note that the simplicity of deriving $\rv(G) \le 2$ critically relies on the orthogonality constraints of query-path vectors \eqref{eq:one-query-sector-orthogonality}, which does not necessarily hold for two-query algorithms.

\begin{lemma}
\label{lem:one-query}
Suppose that the state \(f(x)\) in \eqref{eq:one-query-state} computes a
total Boolean function
    $G:\{+1,-1\}^n\longrightarrow\{+1,-1\}$
exactly.  Then, for some \(\tau\in\{+1,-1\}\), the function \(G\) is one
of
\[
    G(x)=\tau,
    \qquad
    G(x)=\tau x_r,
    \qquad
    G(x)=\tau x_rx_s,
\]
where \(r\neq s\) in the last case.  In particular, \(\rv(G)\leq2\).
Moreover, every relevant
coordinate \(r\) satisfies
    $\norm{f_r}^2=\frac12$.
\end{lemma}

\begin{proof}
Let \(r\) be relevant, and choose a bichromatic \(r\)-edge
\((x,x^{\oplus r})\).  Exactness makes the two endpoint states
orthogonal.  Using \eqref{eq:one-query-sector-orthogonality} and
\(x_0=x_0^{\oplus r}=1\), we obtain
\[
    0
    =\ip{f(x)}{f(x^{\oplus r})}\\
    =\sum_{i=0}^n x_ix_i^{\oplus r}\norm{f_i}^2\\
    =\sum_{i=0}^n\norm{f_i}^2-2\norm{f_r}^2\\
    =1-2\norm{f_r}^2.
\]
Thus,
    $\norm{f_r}^2=\frac12$
    for 
    $r\in\Rel(G)$.
Since the total weight is one, this proves \(\rv(G)\leq2\).

If no coordinate is relevant, then \(G\) is a constant sign.  If \(r\)
is the only relevant coordinate, then \(G\) is a nonconstant Boolean
function of the single sign \(x_r\), and hence
\(G(x)=\tau x_r\) for some \(\tau\).

It remains to consider \(\Rel(G)=\{r,s\}\).  The two relevant sectors
each have weight \(1/2\), so normalization gives
    $f_i=0$
    for 
    $i\in\{0,1,\ldots,n\}\setminus\{r,s\}$.
Consequently,
\[
    f(x)=x_rf_r+x_sf_s,
    \qquad
    f(x^{\oplus r\oplus s})=-f(x).
\]
Here \(x^{\oplus r\oplus s}\) denotes the input obtained by flipping both
coordinates \(r\) and \(s\).
The two states differ only by a global phase and therefore have the same
output.  Hence,
    $G(x^{\oplus r\oplus s})=G(x)$.
Thus, \(G\) is constant on each level set of \(x_rx_s\).  Since both
coordinates are relevant, \(G\) is nonconstant, and
\(G(x)=\tau x_rx_s\) for some \(\tau\in\{+1,-1\}\).
\end{proof}

The next lemma applies to a general normalized degree-one vector-valued
function.  Its coefficient vectors need not initially lie in orthogonal
query sectors.  

\begin{lemma}
\label{lem:degree-one-state}
Let \(h:\{+1,-1\}^n\to\mathcal H\) be the \(\mathcal H\)-valued
function
    $h(x)=h_0+\sum_{i=1}^n x_ih_i$,
and suppose that \(\norm{h(x)}=1\) for every
\(x\in\{+1,-1\}^n\).  Let
\(G:\{+1,-1\}^n\to\{+1,-1\}\) be a Boolean function.  If there are fixed
orthogonal output subspaces \(\mathcal H_+\) and \(\mathcal H_-\) such
that
    $h(x)\in\mathcal H_{G(x)}$
for every \(x\), then \(\QE(G)\leq1\).  Consequently,
\(\rv(G)\leq2\).
\end{lemma}

\begin{proof}
Expanding the constant-norm identity gives
\[
    1=\norm{h(x)}^2
    =\norm{h_0}^2+\sum_{i=1}^n\norm{h_i}^2
    \quad+2\sum_{i=1}^n x_i\operatorname{Re}\ip{h_0}{h_i}
    \quad+2\sum_{1\leq i<j\leq n}
       x_ix_j\operatorname{Re}\ip{h_i}{h_j}.
\]
Fourier uniqueness therefore gives
\begin{equation}
    \operatorname{Re}\ip{h_i}{h_j}=0\quad(i\neq j),
    \qquad
    \sum_{i=0}^n\norm{h_i}^2=1.
    \label{eq:degree-one-real-orthogonality}
\end{equation}

Equation~\eqref{eq:degree-one-real-orthogonality} gives only vanishing
real parts, not necessarily \(\ip{h_i}{h_j}=0\). We consider realification of \(\mathcal H\).  The \emph{conjugate Hilbert
space} \(\overline{\mathcal H}\) is a second copy of \(\mathcal H\).
If \(\overline h\) denotes the copy of \(h\), scalar multiplication and
the inner product are defined by
\[
    \alpha\,\overline h:=\overline{\overline\alpha h},
    \qquad
    \ip{\overline h}{\overline u}_{\overline{\mathcal H}}
    :=\overline{\ip{h}{u}_{\mathcal H}}.
\]
The symbol \(\oplus\) denotes the orthogonal direct sum: its elements are
pairs \(h\oplus\overline u = (h, \overline u)\), and
\[
    \ip{h\oplus\overline u}{h'\oplus\overline{u'}}
    =\ip{h}{h'}+\ip{\overline u}{\overline{u'}}.
\]
Thus, the two components are orthogonal.  For each coefficient, set
\[
    \widetilde h_i
    :=\frac{h_i\oplus\overline{h_i}}{\sqrt2}
    \in\mathcal H\oplus\overline{\mathcal H}
    \qquad(0\leq i\leq n).
\]
Then
\[
    \ip{\widetilde h_i}{\widetilde h_j}
    =\frac12\left(
       \ip{h_i}{h_j}+\overline{\ip{h_i}{h_j}}
      \right)
    =\operatorname{Re}\ip{h_i}{h_j}.
\]
Thus, the nonzero vectors \(\widetilde h_i\) are mutually orthogonal in
the complex Hilbert space \(\mathcal H\oplus\overline{\mathcal H}\), and
their squared norms sum to one.

Let \(e_0,\ldots,e_n\) be orthonormal query-label states, with \(Q_i\)
the projection onto the \(i\)-th label sector, and prepare
\[
    \ket\psi:=\sum_{i=0}^n\norm{\widetilde h_i}\,e_i.
\]
The phase oracle acts as \(e_i\mapsto x_ie_i\), with \(x_0=1\), so
\[
    O_x^Q\ket\psi
    =\sum_{i=0}^n x_i\norm{\widetilde h_i}\,e_i.
\]
Apply a fixed isometry sending each \(e_i\) with \(\widetilde h_i\neq0\) to
\(\widetilde h_i/\norm{\widetilde h_i}\).  The resulting state is
\[
    \widetilde h(x)
    :=\sum_{i=0}^n x_i\widetilde h_i
     =\frac{h(x)\oplus\overline{h(x)}}{\sqrt2},
    \qquad x_0:=1.
\]
Here we use that every \(x_i\in\{+1,-1\}\) is real, so
\(\overline{x_ih_i}=x_i\overline{h_i}\).
The isometry extends to a unitary after enlarging the workspace if
necessary.  For \(\tau\in\{+1,-1\}\), define
    $\widetilde{\mathcal H}_\tau
    :=\mathcal H_\tau\oplus\overline{\mathcal H_\tau}$.
These two output spaces are orthogonal, and
    $\widetilde h(x)\in\widetilde{\mathcal H}_{G(x)}$.
Thus, this is an exact one-query algorithm for \(G\).  Lemma~\ref{lem:one-query}
gives \(\rv(G)\leq2\).
\end{proof}

\subsection{Two-query functions admitting mutually orthogonal query-path}

The simplicity of the one-query argument relies crucially on the mutual
orthogonality of the vectors \(Q_i\ket\phi\).  For two queries, the path
vectors
\[
    v_{ij}=P_iQ_j\ket\phi
\]
are automatically orthogonal when their first indices differ, because
they then lie in different \(P\)-sectors.  For a fixed \(i\), however,
the vectors \(v_{ij}\) with different \(j\) need not be orthogonal.  Their
interference is the main difficulty in the two-query problem.

It is useful first to impose mutual orthogonality on all path vectors.  In
that case, the scalar-weight argument for one query extends immediately
and gives an even stronger bound.

\begin{proposition}
\label{prop:orthogonal-case}
Suppose that the query-path vectors
    $\{v_{ij}:0\leq i,j\leq n\}$
of an exact two-query algorithm for $F$ are pairwise orthogonal, then, \(\rv(F)\leq4\).
\end{proposition}

\begin{proof}
Let \(k\in\Rel(F)\), and choose a bichromatic \(k\)-edge
\((x,x^{\oplus k})\).  Exactness makes the endpoint states orthogonal.
Using the path expansion
    $\ket{\varphi_2(x)}=\sum_{i,j=0}^n x_ix_jv_{ij}$
and pairwise orthogonality, we calculate
\[
    0
    =\ip{\varphi_2(x)}{\varphi_2(x^{\oplus k})}
    =\sum_{i,j=0}^n
       x_ix_jx_i^{\oplus k}x_j^{\oplus k}a_{ij}.
\]
For each pair \((i,j)\),
\[
    x_ix_jx_i^{\oplus k}x_j^{\oplus k}
    =\begin{cases}
       -1,&\text{if exactly one of \(i,j\) equals \(k\)},\\
       +1,&\text{otherwise}.
      \end{cases}
\]
In particular, the \((k,k)\)-term has sign \(+1\).  Recalling
\(\sum_{i,j=0}^n a_{ij}=1\), we obtain
\[
    0
    =\sum_{i,j=0}^n a_{ij}
      -2\left(
       \sum_{\substack{0\leq j\leq n\\j\neq k}}a_{kj}
       +\sum_{\substack{0\leq i\leq n\\i\neq k}}a_{ik}
      \right)\\
    =1-2\left(
       \sum_{\substack{0\leq j\leq n\\j\neq k}}a_{kj}
       +\sum_{\substack{0\leq i\leq n\\i\neq k}}a_{ik}
      \right).
\]
Therefore
\begin{equation}
    \sum_{\substack{0\leq j\leq n\\j\neq k}}a_{kj}
    +\sum_{\substack{0\leq i\leq n\\i\neq k}}a_{ik}
    =\frac12, 
    \qquad \forall\ k\in\Rel(F).
    \label{eq:orthogonal-half-mass}
\end{equation}
Summing over \(k\in\Rel(F)\), since each weight \(a_{ij}\) is counted at most twice, we obtain
    $\frac{\rv(F)}{2}
    \leq2\sum_{i,j=0}^na_{ij}=2$,
which proves \(\rv(F)\leq4\).
\end{proof}

\subsection{Short halfcubes}    \label{sec:halfcubes}

In this section we develop some tools needed for the second step of our proof.

\begin{lemma}
\label{lem:affine-sign}
Let \(G:\{+1,-1\}^n\to\{+1,-1\}\) be a Boolean function.  If
    $G(x)=\alpha_0+\sum_{i=1}^n\alpha_ix_i$,
then \(G\) is a constant sign or a signed literal.
\end{lemma}

\begin{proof}
The pointwise identity \(G(x)^2=1\) expands as
\begin{align*}
    1
    &=\alpha_0^2+\sum_{i=1}^n\alpha_i^2
      +2\sum_{i=1}^n\alpha_0\alpha_i x_i
      +2\sum_{1\leq i<j\leq n}\alpha_i\alpha_jx_ix_j.
\end{align*}
Comparing Fourier coefficients, by the uniqueness of Fourier expansion,
\[
    \alpha_0\alpha_i=0,
    \qquad
    \alpha_i\alpha_j=0\quad(i\neq j),
    \qquad
    \alpha_0^2+\sum_{i=1}^n\alpha_i^2=1.
\]
Hence, exactly one coefficient is nonzero, and it equals \(+1\) or \(-1\).
\end{proof}

\begin{definition}
\label{def:short-halfcube}
A \emph{short halfcube} in $\{+1,-1\}^m$ is a set of one of the forms
\[
    H=\{x:x_r=\tau\}\quad\text{or}\quad H=\{x:x_r=\tau x_s\},
\]
where $r\neq s$ and $\tau\in\{+1,-1\}$.  Equivalently,
$H=\{x:\chi_B(x)=\tau\}$ for a set $B$ of one or two coordinates.
Associated with $H$, choose the affine form
\begin{equation}
    \ell_H(x)
    :=\begin{cases}
       \dfrac{1+\tau x_r}{2},
          &H=\{x:x_r=\tau\},\\[7pt]
       \dfrac{x_r+\tau x_s}{2},
          &H=\{x:x_r=\tau x_s\}.
    \end{cases}
    \label{eq:halfcube-form}
\end{equation}
Then $\ell_H=0$ on $H^c$ and
$|\ell_H|=1$ on $H$.
\end{definition}

For example, take $H=\{x\in\{+1,-1\}^m:x_1=-x_2\}$.  Eliminating
$x_2$ by the relation $x_2=-x_1$ parametrizes $H$ by the ordinary cube
$\{+1,-1\}^{m-1}$ on the free coordinates $(x_1,x_3,\ldots,x_m)$.
The corresponding bijection is
    $(x_1,x_3,\ldots,x_m)\mapsto(x_1,-x_1,x_3,\ldots,x_m)$.

\begin{lemma}
\label{lem:affine-support}
Let $f:\{+1,-1\}^n\to V$, $f(x)=v_0+\sum_{i=1}^n x_iv_i$, be a
nonzero affine vector-valued function, where $V$ is a real or complex
inner-product space.  Then,
\[
    |\supp(f)|\geq2^{n-1}.
\]
If equality holds, then $\supp(f)$ is a short halfcube.
\end{lemma}

\begin{proof}
If $v_i=0$ for every $i\in[n]$, then $f$ is a nonzero constant and
its support is the full cube.  Otherwise, after relabeling coordinates,
assume $v_1\neq0$.

Define 
$\lambda(w):=\operatorname{Re}\ip{v_1}{w}$ for $w\in V$. 
Then,  $c:=\lambda(v_1)=\norm{v_1}^2>0$.
Consider the function 
\[
    g(x) :=\lambda(f(x))=a(x_2,\ldots,x_n)+cx_1,
    \qquad
    a(x_2,\ldots,x_n):=\lambda(v_0)+\sum_{i=2}^n x_i\lambda(v_i).
\]
For every fixed $x'=(x_2,\ldots,x_n)$, the two values $g(+1,x')=a(x')+c$ and $g(-1,x')=a(x')-c$ cannot both be zero, as otherwise one would have $c = 0$, a contradiction. This shows 
    $|\supp(g)| \geq2^{n-1}$.
Since $g(x)=\lambda(f(x))$, we have
$\supp(g)\subseteq\supp(f)$, and hence $|\supp(f)|\geq2^{n-1}$.

Now suppose that $|\supp(f)|=2^{n-1}$, the analysis above implies 
    $|\supp(g)|=2^{n-1}$, 
and exactly one of
$a(x')+c$ and $a(x')-c$ vanishes.  Equivalently,
$G(x'):=a(x')/c$ lies in $\{+1,-1\}$ for every
$x'\in\{+1,-1\}^{n-1}$.  The function $G$ is affine and sign-valued.
By Lemma~\ref{lem:affine-sign}, for some $\tau\in\{+1,-1\}$, either
$G(x')=\tau$ or $G(x')=\tau x_r$ for some $r\in\{2,\ldots,n\}$.
Since $g(x)=c(G(x')+x_1)$, the two cases give, respectively,
\[
    \supp(f)=\supp(g)=\{x:x_1=\tau\}
    \quad\text{or}\quad
    \supp(f)=\supp(g)=\{x:x_1=\tau x_r\}.
\]
Both are short halfcubes.
\end{proof}

\begin{lemma}
\label{lem:halfcube-factor}
Let $H\subseteq\{+1,-1\}^m$ be a short halfcube.  Suppose that
$f(x,y)$ is a scalar- or vector-valued multilinear function of
degree at most two which vanishes whenever $x\notin H$, for every value
of the additional Boolean variables $y$.  Then
\[
    f(x,y)=\ell_H(x)\,h(x,y).
\]
Here $h$ is scalar-valued when $f$ is scalar-valued and vector-valued
when $f$ is vector-valued, and $h$ has degree at most one on the full
Boolean cube.
\end{lemma}

\begin{proof}
First suppose $H=\{x:x_r=\tau\}$.  Write $f=a+x_rb$, where $a,b$ do
not depend on $x_r$.  The degree bound gives $\deg(b)\leq1$.
Vanishing at $x_r=-\tau$ gives $a=\tau b$,
and therefore
\[
    f=\frac{1+\tau x_r}{2}(2\tau b)=\ell_H(2\tau b).
\]
Thus, one may take $h=2\tau b$, which has degree at most one on the
full cube.

Now suppose $H=\{x:x_rx_s=\tau\}$.  Write
$f=a+x_rb+x_sc+x_rx_sd$,
where the four coefficients are independent of $x_r,x_s$.  Here
$\deg(b),\deg(c)\leq1$, while $d$ is constant.  On $H^c$ one has
$x_s=-\tau x_r$.  Vanishing there gives $a=\tau d$ and $b=\tau c$.
Consequently,
\[
    f=\frac{x_r+\tau x_s}{2}\,2\tau(c+dx_r)
     =\ell_H\,2\tau(c+dx_r).
\]
Thus, one may take $h=2\tau(c+dx_r)$, which again has degree at most
one on the full cube.  
\end{proof}

\section{Proof of the main theorem}
\label{sec:main-proof}

We now prove Theorem~\ref{thm:main}, i.e., its upper bound.  The upper bound has two steps.
First, we show that any hypothetical counterexample has at least three
relevant coordinates outside the first-query support.  Second, 
we show that three such coordinates forces exactly two types of second-query branches, which can be used to prove the function depends on at most six variables. 

Technically, in the second part, we establish a key orthogonality property (the in particular part in Lemma \ref{lem:type-orthogonality}), which allows us to describe each type as a support of an affine map. We then use the properties of affine map established in Section \ref{sec:halfcubes} to conclude exactly two types occur. This also allows us to use the short halfcube geometry, i.e., shoft halfcube can be described by a character $\chi_B$ for some $B \subseteq S$ of size $|B| \le 2$. We then analyze subfunctions after fixing assignments on $B$, and conclude by using similar analysis for proving $J_E(1)=2$.

\subsection{Step 1: Bounding the first-query support}
\label{sec:first-support}

We now obtain information that ordinary degree alone does not provide:
each coordinate in the first-query support lowers the degree of a toggled
copy of the output function.

The key move is to apply the second oracle $O_x^P$ once more after the
exactness identity.  Since $(O_x^P)^2=I$, this \emph{unqueries} the
second query on the right-hand side and exposes the first-query state.

\begin{lemma}
\label{lem:first-query-degree}
For the exact two-query representation
$O_x^PO_x^Q\ket\phi$, every coordinate $j\in S$ satisfies
\[
    \deg(x_jF)\leq3,
\]
where $(x_jF)(x):=x_jF(x)$.
\end{lemma}

\begin{proof}
Recall that exactness says
\[
    \mathsf M O_x^PO_x^Q\ket\phi
    =F(x)O_x^PO_x^Q\ket\phi.
\]
Multiplying on the left by $O_x^P$ and using $(O_x^P)^2=I$ gives
\begin{equation}
    O_x^P\mathsf M O_x^PO_x^Q\ket\phi
    =F(x)O_x^Q\ket\phi.
    \label{eq:transported-exactness}
\end{equation}
The left-hand side is a vector-valued polynomial of degree at most three.
Put $\omega_j:=Q_j\ket\phi$.  The vectors
$\omega_0,\ldots,\omega_n$ are mutually orthogonal and
$\norm{\omega_j}^2=q_j$.  Taking the inner product of
\eqref{eq:transported-exactness} with $\omega_j$, its right-hand side becomes
\[
    F(x)\ip{\omega_j}{\sum_{\ell=0}^n x_\ell\omega_\ell}
    =q_jx_jF(x).
\]
The left-hand side remains a scalar polynomial of degree at most three.
Since $j\in S$, one has $q_j>0$, and division by $q_j$ proves the
claim.
\end{proof}

For later use, we recall the coefficient form of the Fourier toggle:
\begin{equation}
    x_jF(x)
    =\sum_{A\subseteq[n]}\widehat F(A)
       \chi_{A\mathbin\triangle\{j\}}(x),
    \qquad
    \widehat{x_jF}(B)
    =\widehat F(B\mathbin\triangle\{j\}).
    \label{eq:F-Fourier-toggle}
\end{equation}

\begin{proposition}
\label{prop:support-four}
If $Q_E(F) \le 2$ and
$\rv(F)\geq5$, then every $A$ satisfying $\widehat F(A)\neq0$ and 
$|A|\geq3$ contains the first-query support $S$.  In particular,
$|S|\leq4$.
\end{proposition}

\begin{proof}
The Nisan--Szegedy bound~\cite{NisanSzegedy1994}
$\rv(F)\leq d2^{d-1}$, where $d=\deg(F)$, implies that
$\deg(F)\leq2$ entails $\rv(F)\leq2\cdot2^{2-1}=4$.  Hence,
$\rv(F)\geq5$ gives $\deg(F)\geq3$.  Fix $A\subseteq[n]$ with
$\widehat F(A)\neq0$ and $|A|\geq3$, and fix $j\in S$.
Lemma~\ref{lem:first-query-degree} and
\eqref{eq:F-Fourier-toggle} give
\[
    \widehat{x_jF}(A\mathbin\triangle\{j\})
    =\widehat F(A)\neq0,\qquad |A\mathbin\triangle\{j\}|\leq3.
\]
If $j\notin A$, then $|A\mathbin\triangle\{j\}|=|A|+1\geq4$,
contradicting $|A\mathbin\triangle\{j\}|\leq3$.  Hence $j\in A$,
and therefore $S\subseteq A$.

Finally, $\deg(F)\leq 2Q_E(F) \le 4$. Hence, every $A$ satisfying $\widehat F(A)\neq0$ has size $|A| \le 4$. So, $|S|\leq 4$.
\end{proof}

Thus, in the regime $\rv(F)\geq5$, at most four physical coordinates
have positive first-query weight.  This gives the starting point for the
next step.

\begin{corollary}
\label{cor:three-outside}
If $Q_E(F) \le 2$  and
$\rv(F)\geq7$, then at least three relevant coordinates lie outside
the first-query support.
\end{corollary}

\begin{proof}
Proposition~\ref{prop:support-four} gives $|S|\leq4$.  Therefore,
$|\Rel(F)\setminus S|\geq\rv(F)-|S|\geq7-4=3$.
\end{proof}

\subsection{Step 2: The structure of two-query algorithms}
\label{sec:two-query-structure}

We now prove the complementary structural statement that supplies the
second step of the argument.

\begin{theorem}
\label{thm:three-outside}
Suppose that an exact two-query algorithm for $F$ has at least three
relevant coordinates outside its first-query support $S$.  Then,
    $\rv(F)\leq6$.
\end{theorem}

%\begin{remark}
%One can in fact prove the stronger statement that a single relevant coordinate outside the first-query support already implies $\rv(F)\leq6$.  That proof must analyze how one nonconstant branch type can coexist with constant branches and requires a finer case distinction. We omit that longer argument because it is not needed here.
%\end{remark}

\subsubsection{One-query branches to types}

Fix an assignment $\varepsilon=(\varepsilon_j)_{j\in S}
\in\{+1,-1\}^{S}$ to the coordinates in the first-query support, and
let $y=(y_i)_{i\in[n]\setminus S}\in\{+1,-1\}^{[n]\setminus S}$
denote the remaining signs.  We write $(\varepsilon,y)$ for the
corresponding full input.  Since $Q_j\ket\phi=0$ for $j\notin S$,
the state after the first query is
\begin{equation}
    \ket{\varphi_1(\varepsilon,y)}
    =Q_0\ket\phi+\sum_{j\in S}\varepsilon_jQ_j\ket\phi.
    \label{eq:fixed-first-state}
\end{equation}
Since  this is independent of $y$, by an abuse of notation, we denote
$\ket{\varphi_1(\varepsilon)}:=\ket{\varphi_1(\varepsilon,y)}$.

For every second-query label $i\in\{0,1,\ldots,n\}$, define the affine
vector-valued function
\begin{equation}
    s_i(\varepsilon):=P_i\ket{\varphi_1(\varepsilon)}
    =v_{i0}+\sum_{j\in S}\varepsilon_jv_{ij}.
    \label{eq:sector-map}
\end{equation}
We also define the vector-valued function
$c(\varepsilon):=s_0(\varepsilon)+\sum_{i\in S}\varepsilon_is_i(\varepsilon)$.
Writing $x=(\varepsilon,y)$ and expanding the sectors of the second
query gives
\begin{equation}
    \ket{\varphi_2(\varepsilon,y)}
    =O_x^P\ket{\varphi_1(\varepsilon)}
    =c(\varepsilon)+\sum_{i\in[n]\setminus S}y_is_i(\varepsilon).
    \label{eq:branch-decomposition}
\end{equation}

We record the cross-assignment sector orthogonality used below.  For any
$\varepsilon,\varepsilon'\in\{+1,-1\}^S$ and any distinct
$i,j\in\{0,1,\ldots,n\}$,
\begin{equation}
    \ip{s_i(\varepsilon)}{s_j(\varepsilon')}
    =\ip{\varphi_1(\varepsilon)}{P_iP_j\varphi_1(\varepsilon')}=0.
    \label{eq:cross-sector-orthogonality}
\end{equation}
Moreover, $c(\varepsilon)$ lies in the direct sum of the $P_0$-sector
and the $P_i$-sectors for $i\in S$.  Hence, for
$k\in[n]\setminus S$,
\begin{equation}
    \ip{c(\varepsilon)}{s_k(\varepsilon')}
    =\ip{s_0(\varepsilon)}{s_k(\varepsilon')}
      +\sum_{i\in S}\varepsilon_i
        \ip{s_i(\varepsilon)}{s_k(\varepsilon')}=0.
    \label{eq:c-s-orthogonality}
\end{equation}
These identities hold for both $\varepsilon = \varepsilon'$ and $\varepsilon \neq \varepsilon'$.  
In particular, normalization of the final state $\ket{\varphi_2(\varepsilon,y)}$ gives
\[
    \norm{c(\varepsilon)}^2+\sum_{i\in[n]\setminus S}\norm{s_i(\varepsilon)}^2=1.
\]

For fixed $\varepsilon$, equation~\eqref{eq:branch-decomposition} is
an exact one-query state in the variables $y$.  Indeed,
$c(\varepsilon)$ is orthogonal to every $s_i(\varepsilon)$ with
$i\notin S$, and these latter vectors are mutually orthogonal.  Thus,
$c(\varepsilon)$ serves as the dummy-sector component for the reduced
query.  The computed restricted function is
$F_\varepsilon(y):=F(\varepsilon,y)$.
Lemma~\ref{lem:one-query} therefore gives
\begin{equation}
    F_\varepsilon(y)
    =\tau_\varepsilon\chi_{T_\varepsilon}(y),
    \qquad
    |T_\varepsilon|\leq2,
    \label{eq:branch-type}
\end{equation}
where $\tau_\varepsilon:=F_\varepsilon(\1)\in\{+1,-1\}$
is its overall sign, and $T_\varepsilon\subseteq[n]\setminus S$ is
exactly the set of relevant coordinates of $F_\varepsilon$.  We call
$T_\varepsilon$ the \emph{type} of the branch.  The empty type means
that the branch is constant; a nonempty type is called
\emph{nonconstant}.  The weight conclusion in
Lemma~\ref{lem:one-query} also gives
\begin{equation}
    \norm{s_i(\varepsilon)}^2=\frac12, 
    \qquad \forall\ i\in T_\varepsilon.
    \label{eq:branch-half-weight}
\end{equation}

\subsubsection{At most two nonconstant types}

\begin{lemma}
\label{lem:type-orthogonality}
If $T_\varepsilon\neq T_{\varepsilon'}$, then
\[
    \ip{\varphi_2(\varepsilon,y)}
       {\varphi_2(\varepsilon',y')}=0
    \qquad\text{for all }y,y'.
\]
In particular,
\[
    \ip{s_i(\varepsilon)}{s_i(\varepsilon')}=0
    \qquad \forall\ i\in[n]\setminus S.
\]
This includes the case in which exactly one branch is constant.
\end{lemma}

\begin{proof}
The cross-sector identities
\eqref{eq:cross-sector-orthogonality} and
\eqref{eq:c-s-orthogonality} give
\begin{equation}
    \kappa(y,y')
    :=\ip{\varphi_2(\varepsilon,y)}
            {\varphi_2(\varepsilon',y')}
    =\ip{c(\varepsilon)}{c(\varepsilon')}
      +\sum_{i\in[n]\setminus S}y_iy_i'
          \ip{s_i(\varepsilon)}{s_i(\varepsilon')}.
    \label{eq:cross-gram}
\end{equation}
Thus, the scalar-valued function $\kappa(y,y')$ depends only on the
coordinatewise product $w=y\odot y'$, where $w_i=y_iy_i'$.

Fix $w\in\{+1,-1\}^{[n]\setminus S}$ and put $y'=y\odot w$.
Equation~\eqref{eq:branch-type} gives
\[
    F_\varepsilon(y)F_{\varepsilon'}(y\odot w)
    =\tau_\varepsilon\tau_{\varepsilon'}
     \chi_{T_\varepsilon\mathbin\triangle T_{\varepsilon'}}(y)
     \chi_{T_{\varepsilon'}}(w).
\]
Since $T_\varepsilon\neq T_{\varepsilon'}$, their symmetric difference
is nonempty, and its character takes both signs.  We may therefore choose
$y$ so that $F_\varepsilon(y)F_{\varepsilon'}(y\odot w)=-1$.
For this choice, $F(\varepsilon,y)\neq F(\varepsilon',y\odot w)$.
Exactness then implies
\[
    \kappa(y,y\odot w)
    =\ip{\varphi_2(\varepsilon,y)}{\varphi_2(\varepsilon',y\odot w)}=0.
\]
The expression in \eqref{eq:cross-gram} depends only on $w$, not on the
chosen $y$.  It therefore vanishes for every $w$, so
$\kappa(y,y')=0$ for all $y,y'$.  Finally, as a function of $w$,
\eqref{eq:cross-gram} is
\[
    \ip{c(\varepsilon)}{c(\varepsilon')}
    +\sum_{i\in[n]\setminus S}w_i
       \ip{s_i(\varepsilon)}{s_i(\varepsilon')}.
\]
Fourier uniqueness forces the coefficient of every $w_i$ to vanish,
which proves the final assertion.
\end{proof}

For a nonconstant type $T\subseteq[n]\setminus S$, define its
\emph{type cell}
\begin{equation}
    \mathcal C_T
    :=\{\varepsilon\in\{+1,-1\}^{S}:T_\varepsilon=T\}
      \subseteq\{+1,-1\}^{S}.
    \label{eq:type-cell}
\end{equation}
Note that $\mathcal C_T \cap \mathcal C_{T'} = \emptyset$ if $T\neq T'$.
The density of $\mathcal C_T$ is $|\mathcal C_T|/2^{|S|}$.
We say that $T$ \emph{occurs} if $\mathcal C_T\neq\varnothing$.

\begin{lemma}
\label{lem:large-type-cells}
Every occurring nonconstant type has a type cell of density at least
$1/2$ in $\{+1,-1\}^{S}$.  Consequently, at most two nonconstant
types occur.  If two occur, their cells are complementary short halfcubes
and there are no constant branches.
\end{lemma}

\begin{proof}
We will construct, for each occurring nonconstant type $T$, an
affine vector-valued function whose support is exactly $\mathcal C_T$.
Lemma~\ref{lem:affine-support} will then bound the size of the cell.  If
$S=\varnothing$, there is only one branch, so all conclusions are
immediate.  Assume $S\neq\varnothing$.

Fix $i\in T$, and let
\[
    \mathcal V_{i,T}:=\Span\{s_i(\varepsilon):\varepsilon\in\mathcal C_T\}
\]
Let $\Pi_{i,T}$ be the orthogonal projection onto the Hilbert subspace
$\mathcal V_{i,T}$.   Consider
\[
    h_{i,T}(\varepsilon):=\Pi_{i,T}s_i(\varepsilon).
\]
We consider two case.

If $\varepsilon\in\mathcal C_T$, then
$s_i(\varepsilon)\in\mathcal V_{i,T}$, and therefore
\[
    h_{i,T}(\varepsilon)=s_i(\varepsilon),\qquad
    \norm{h_{i,T}(\varepsilon)}^2=\norm{s_i(\varepsilon)}^2=\frac12.
\]
In particular, $h_{i,T}(\varepsilon)\neq0$.

If $\varepsilon\notin\mathcal C_T$, then $T_\varepsilon\neq T$.
For every $\varepsilon'\in\mathcal C_T$,
Lemma~\ref{lem:type-orthogonality} gives
$\ip{s_i(\varepsilon)}{s_i(\varepsilon')}=0$.  Hence,
\[
    s_i(\varepsilon)\perp
    \Span\{s_i(\varepsilon'):\varepsilon'\in\mathcal C_T\}
    =\mathcal V_{i,T},
\]
and $h_{i,T}(\varepsilon)=\Pi_{i,T}s_i(\varepsilon)=0$.

Combining the two cases gives
\[
    h_{i,T}(\varepsilon)
    =\begin{cases}
       s_i(\varepsilon),&\varepsilon\in\mathcal C_T,\\
       0,&\varepsilon\notin\mathcal C_T.
      \end{cases}
\]
The value in the first case is nonzero.  Thus,
$\supp(h_{i,T})=\mathcal C_T$, and Lemma~\ref{lem:affine-support} yields
$|\mathcal C_T|\geq2^{|S|-1}$, and furthermore, when equality holds, Lemma~\ref{lem:affine-support}  implies that $\mathcal C_T$ is a short halfcube. The rest of the statement follows from the disjointness of different $\mathcal C_T$, as noted after the definition. 
\end{proof}

\subsubsection{The case of two nonconstant types}

\begin{proposition}
\label{prop:two-type-compression}
If two distinct nonconstant branch types occur, then
    $\rv(F)\leq6$.
\end{proposition}

\begin{proof}
By Lemma~\ref{lem:large-type-cells}, the two type cells are complementary
short halfcubes $H,H^c\subseteq\{+1,-1\}^{S}$, and every branch is
nonconstant.  Put $Y:=\{+1,-1\}^{[n]\setminus S}$.
Define the two Hilbert subspaces
\begin{align*}
    \mathcal K_H&:=\Span\{\ket{\varphi_2(\varepsilon,y)}:
       \varepsilon\in H,\ y\in Y\},\\
    \mathcal K_{H^c}&:=\Span\{\ket{\varphi_2(\varepsilon,y)}:
       \varepsilon\in H^c,\ y\in Y\}.
\end{align*}
Lemma~\ref{lem:type-orthogonality} gives
$\mathcal K_H\perp\mathcal K_{H^c}$.  
Let $\Pi_{\mathcal K_H}$ and $\Pi_{\mathcal K_{H^c}}$ denote the respective  orthogonal projections. Define the projected state maps
\[
    \psi_H(\varepsilon,y):=\Pi_{\mathcal K_H}\ket{\varphi_2(\varepsilon,y)},
    \qquad
    \psi_{H^c}(\varepsilon,y):=\Pi_{\mathcal K_{H^c}}\ket{\varphi_2(\varepsilon,y)}.
\]
Since the state belongs to $\mathcal K_H$ on $H$ and to the
orthogonal subspace $\mathcal K_{H^c}$ on $H^c$,
\[
    \psi_H(\varepsilon,y)
    =\begin{cases}
       \ket{\varphi_2(\varepsilon,y)},&\varepsilon\in H,\\
       0,&\varepsilon\in H^c.
      \end{cases}
\]
The analogous formula holds for $\psi_{H^c}$ with $H$ and $H^c$
interchanged.  The path expansion
$\ket{\varphi_2(x)}=\sum_{i,j=0}^n x_ix_jv_{ij}$
has degree at most two, and a fixed linear projection does not increase
degree.  Lemma~\ref{lem:halfcube-factor} therefore gives
\begin{equation}
    \psi_H(\varepsilon,y)=\ell_H(\varepsilon)h_H(\varepsilon,y),
    \qquad
    \psi_{H^c}(\varepsilon,y)=\ell_{H^c}(\varepsilon)h_{H^c}(\varepsilon,y),
    \label{eq:projected-factorization}
\end{equation}
where $h_H$ and $h_{H^c}$ are vector-valued functions of degree at
most one on the full cube $\{+1,-1\}^{S}\times Y$.

For $\varepsilon\in H$, one has $\ell_H(\varepsilon)^2=1$, and
\eqref{eq:projected-factorization} gives
\[
    h_H(\varepsilon,y)=\ell_H(\varepsilon)\ket{\varphi_2(\varepsilon,y)}.
\]
Consequently, $\norm{h_H(\varepsilon,y)}=1$.  Since
$\ell_H(\varepsilon)$ is a real sign, exactness also gives
\begin{equation}    \label{eq:restricted_function_exact_calc}
    \mathsf M h_H(\varepsilon,y)
    =\ell_H(\varepsilon)\mathsf M\ket{\varphi_2(\varepsilon,y)}
    =\ell_H(\varepsilon) F(\varepsilon,y) \ket{\varphi_2(\varepsilon,y)}
    =F(\varepsilon,y)h_H(\varepsilon,y),
    \qquad
    \forall\ \varepsilon \in H.
\end{equation}

We now fix all coordinates defining the two halfcubes.  Write
\[
    H=\{\varepsilon:\chi_B(\varepsilon)=\tau\},
    \qquad
    H^c=\{\varepsilon:\chi_B(\varepsilon)=-\tau\},
\]
where $B\subseteq S$, $1\leq|B|\leq2$, and
$\tau\in\{+1,-1\}$.  Put $R:=[n]\setminus B$.  For each
$\beta\in\{+1,-1\}^{B}$, define the restriction
\[
    F_\beta(z):=F(\beta,z),\qquad z\in\{+1,-1\}^{R},
\]
where $(\beta,z)$ denotes the full input with the indicated coordinate
assignments.  If $\chi_B(\beta)=\tau$, the corresponding first-query
assignment lies in $H$ for every $z$.  The preceding calculation shows
that
\[
    h_{H,\beta}(z):=h_H(\beta,z)
       =a_H(\beta)+\sum_{j\in R}z_ja_{H,j}
\]
is a normalized state of degree at most one which exactly computes
$F_\beta$.  Thus, Lemma~\ref{lem:degree-one-state} gives
$\QE(F_\beta)\leq1$.

Crucially, since $h_H$ is affine on the full cube, the coefficient
vectors $a_{H,j}$ are independent of $\beta$; only the constant term
$a_H(\beta)$ changes.  Normalization and Parseval give
\[
    \norm{a_H(\beta)}^2+\sum_{j\in R}\norm{a_{H,j}}^2=1.
\]
Let
\[
    R_H:=\bigcup_{\substack{\beta\in\{+1,-1\}^{B}\\
                            \chi_B(\beta)=\tau}}\Rel(F_\beta).
\]
For $j\in R_H$, choose a compatible $\beta$ and a bichromatic $j$-edge
with endpoints $z,z^{\oplus j}$ for $F_\beta$.  The two corresponding
states are orthogonal unit vectors, and
\[
    h_{H,\beta}(z)-h_{H,\beta}(z^{\oplus j})=2z_ja_{H,j}.
\]
Consequently,
\[
    2=\norm{h_{H,\beta}(z)-h_{H,\beta}(z^{\oplus j})}^2
      =4\norm{a_{H,j}}^2.
\]
Every coordinate in $R_H$ therefore has the same fixed weight $1/2$,
regardless of which compatible assignment witnesses its relevance.
Since $\sum_{j\in R}\norm{a_{H,j}}^2\leq1$, we obtain $|R_H|\leq2$.
The identical argument for $h_{H^c}$ gives
\[
    R_{H^c}:=\bigcup_{\substack{\beta\in\{+1,-1\}^{B}\\
                                \chi_B(\beta)=-\tau}}\Rel(F_\beta),
    \qquad |R_{H^c}|\leq2.
\]

Finally, a bichromatic edge in a coordinate outside $B$ has the same
assignment $\beta$ at both endpoints.  Thus, every such relevant
coordinate belongs to $R_H$ or $R_{H^c}$, and
\[
    \rv(F)\leq|B|+|R_H|+|R_{H^c}|\leq2+2+2=6.   \qedhere
\]
\end{proof}

\subsubsection{Proof of Theorem~\ref{thm:three-outside}}

\begin{proof}[Proof of Theorem~\ref{thm:three-outside}]
Let $Z:=\Rel(F)\setminus S$.
By hypothesis, $|Z|\geq3$.  Fix $j\in Z$, and choose a bichromatic
$j$-edge of the full cube.  Its two endpoints have the same restriction
$\varepsilon$ to $S$.  The corresponding one-query branch
$F_\varepsilon$ therefore depends on $y_j$, so
$j\in T_\varepsilon$.  It follows that
\[
    Z\subseteq
    \bigcup_{\varepsilon\in\{+1,-1\}^{S}}T_\varepsilon.
\]
Every type has size at most two.  Since $|Z|\geq3$, at least two
distinct nonconstant types occur.  The fixed-$B$ argument in the proof of
Proposition~\ref{prop:two-type-compression} supplies $B\subseteq S$ with
$|B|\leq2$ and two families of restrictions, each involving in total at
most two relevant coordinates outside $B$.  Hence,
$\rv(F)\leq|B|+2+2\leq6$.
\end{proof}

\subsection{Completing the proof}
\label{sec:completing-proof}

\begin{proof}[Proof of Theorem~\ref{thm:main}]
Suppose for contradiction that $\QE(F)\leq2$ and
$\rv(F)\geq7$.  By padding a missing query in the dummy sector, we may
use the two-query representation of Section~\ref{sec:prelim}.
Corollary~\ref{cor:three-outside} provides at least three relevant
coordinates outside the first-query support.  Theorem~\ref{thm:three-outside}
then gives $\rv(F)\leq6$, a contradiction to $\rv(F)\geq7$.
\end{proof}

\bibliographystyle{alpha}
\bibliography{two_exact_queries}

\end{document}